\documentclass[12pt]{article}
\usepackage{amscd,amsmath,amssymb,amstext,amsthm,exscale,latexsym}
\usepackage{graphicx}
\newcommand {\al}   {\alpha}       
       
\newcommand {\dl}   {\delta}       
        
\newcommand {\ve}   {\varepsilon}

         \newcommand {\om}  {\omega}

\newcommand {\pl}   {\partial}     
\renewcommand {\sin}{{\sf\,sin\,}}       \renewcommand {\cos}{{\sf\,cos\,}}
         \newcommand {\ctg}{{\sf\,ctg\,}}

       \renewcommand {\lim}{{\sf\,lim\,}}

     \newcommand   {\diag}{{\sf\,diag\,}}

   \newcommand {\MB}  {{\mathbb B}}

\newcommand {\MM}  {{\mathbb M}}   
\newcommand {\MO}  {{\mathbb O}}   
   \newcommand {\MR}  {{\mathbb R}}
\newcommand {\MS}  {{\mathbb S}}   
\newcommand {\MU}  {{\mathbb U}}

\newcommand {\Go}  {\mathfrak{o}}   
   
\newcommand {\Gs}  {\mathfrak{s}}

\newcommand {\one}  {1\!\!1}

\newtheorem{theorem}{Theorem}[section]

\theoremstyle{definition}

\begin{document}
\title     {Gauge parameterization of the $n$-field}
\author    {M. O. Katanaev
            \thanks{E-mail: katanaev@mi.ras.ru}\\ \\
            \sl Steklov mathematical institute,\\
            \sl 119991, Moscow, ul.~Gubkina, 8.\\
            \sl N.~I.~Lobachevsky Institute of Mathematics and Mechanics,\\
            \sl Kazan Federal University,\\
            \sl ul.~Ktremlevskaya 35, Kazan, 420008 Russia}
\maketitle
\begin{abstract}
We propose gauge parameterization of the three-dimensional $n$-field using
orthogonal $\MS\MO(3)$-matrix which, in turn, is defined by a field
taking values in the Lie algebra $\Gs\Go(3)$ (rotational-angle field). The
rotational-angle field has an additional degree of freedom, which corresponds to
the gauge degree of freedom of rotations around the $n$-field. As a result, we
obtain a gauge model with local $\MS\MO(2)\simeq\MU(1)$ symmetry that does not
contain a $\MU(1)$ gauge field.
\end{abstract}
%******************************************************************************
\section{Introduction}
%*******************************************************************************
Gauge models are an essential part of modern mathematical physics. The gauge
invariance of Yang--Mills models is achieved by introducing gauge fields which
are components of local connection form for the corresponding principal fiber
bundle (see., e.g., \cite{SlaFad88}). It is these models that are usually called
gauge models. In the present paper, the gauge model is understood in a wider
sense: it is any field model that is invariant under some local transformation
group whose parameters can depend sufficiently smooth on a space-time point. In
this sense, general relativity is also a gauge model, because the
Hilbert--Einstein action is invariant with respect to general coordinate
transformations parameterized by four arbitrary functions. In addition, the
action depends only on the metric or vierbein, which are not gauge fields in the
strict sense.

Thus, the models invariant under local transformations do not always contain
gauge fields. In the present paper, we construct a new class of models with
local $\MU(1)\simeq\MS\MO(2)$ invariance that does not include a gauge
$\MU(1)$-field. This model arose in the geometric theory of defects [2--6].
\nocite{KatVol92,KatVol99,Katana05,Katana13B,Katana17C}
Namely, some continuous medium possesses a spin structure in addition to elastic
properties. For instance, the ferromagnetic properties of media are described by
the distribution of magnetic moments. In the continuum approximation, such a
medium is considered as a three-dimensional manifold $\MM\approx\MR^3$ with
given unit vector field $n(x):~\MM\to\MS^2$ that describes the spin distribution
in the medium. If the unit vector field is sufficiently smooth, then we say that
the spin structure has no defects and write down some Lagrangian for $n$-field.
However, in nature, the spin structure often contains defects, which are called
disclinations. These are any discontinuities and other singularities of the
$n$-field, whose supports can be located at points, on lines, or on surfaces.
If there are few disclinations, then we can pose a problem for the $n$-field
outside defects with appropriate boundary conditions at the discontinuities of
the $n$-field. This approach is applicable to a small number of separate
disclinations. However, if there are many disclinations (which is the most
common case for real media), the boundary conditions become so complicated that
one cannot hope to solve the corresponding boundary value problems. In the
limiting case of continuous distribution of disclinations, the $n$-field has
discontinuities at every point, which means that it does not exist at all.
Therefore, the $n$-field is not suitable for describing media with
disclinations, and we need a new formalism.

In order to describe single disclinations as well as their continuous
distribution, the geometric theory of defects was proposed [2--6]. In this
approach, the $n$-field is substituted by a new variable, an
$\MS\MO(3)$-connection, which is nonsingular for continuous distribution of
disclinations. For single disclinations it may have singularities at points,
on lines, or on surfaces. The new variable is introduced as follows. We fix some
direction in space and parameterize the $n$-field by an orthogonal matrix.
In turn, the rotation matrix is parameterized by an element of the Lie algebra
$\Gs\Go(3)$; i.e.\ we have a rotational-angle field
$\boldsymbol{\om}(x):~\MM\to\Gs\Go(3)$. If there are no disclinations, then the
rotational-angle field $\boldsymbol{\om}(x)$ is a smooth function and the
partial derivatives $\pl_\mu\boldsymbol{\om}$ exist. In the presence of
disclinations, the partial derivatives may not exist, and we introduce a new
variable $\pl_\mu\boldsymbol{\om}\mapsto\om_\mu{}$, which is a 1-form with
values in the Lie algebra $\Gs\Go(3)$ and which is identified with the
components of a local $\MS\MO(3)$-connection form. In this case, disclinations
exist if and only if the curvature tensor for the $\MS\MO(3)$-connection
is nonzero. On simply connected domains with zero curvature tensor, the
$\MS\MO(3)$-connection is a pure gauge and one can construct the the
rotational-angle field $\boldsymbol{\om}$ and the $n$-field. In the other
cases, the rotational-angle and $n$-field do not exist, as it should be, for
example, for a continuous distribution of disclinations.

The change of variables $n(x)\mapsto\boldsymbol{\om}(x)$ is a necessary
attribute of the geometric theory of defects and thus needs to be carefully
analysed. The problem is that this change of variables is not one-to-one:
the $n$-field has two degrees of freedom because of the condition $n^2=1$, and
the rotational-angle field $\boldsymbol{\om}$ has three degrees of freedom. The
additional degree of freedom corresponds to $\MS\MO(2)$-rotations around the
$n$-field and is a gauge one. This question is the subject of the present paper.
%******************************************************************************
\section{Angle parameterization of the $n$-field}
%*******************************************************************************
In the geometric theory of defects, a unit vector field $n(x):~\MR^3\to\MS^2$,
which describes, for example, the distribution of magnetic moments in
ferromagnets, is parameterized by the rotational-angle field. To this end, we
fix some direction in Euclidean space $\MR^3$ by choosing a unit vector $n_0$.
Then the unit vector field is uniquely represented by the orthogonal matrix:
\begin{equation}                                                  \label{enpara}
  n^i(x):=n^j_0 S_j{}^i\big(\boldsymbol{\om}(x)\big),\qquad S_j{}^i\in\MO(3).
\end{equation}
In turn, the matrix is uniquely parameterized by an element
$\boldsymbol{\om}(x)=\big(\om^i(x)\big)$ (rotation-angle vector field) of the
Lie algebra $\Gs\Go(3)$. The rotational-angle $\boldsymbol{\om}$ parameterizes
the proper rotation subgroup $\MS\MO(3)\subset\MO(3)$ as follows. The direction
of the vector $\boldsymbol{\om}$ coincides with the rotational axis, and its
length is equal to the rotational angle. For definiteness, we assume that the
rotation angle varies in the range $|\boldsymbol{\om}|\le\pi$. Then the end of
$\boldsymbol{\om}$ runs over all points of the closed ball
$\bar\MB^3_\pi(0)\hookrightarrow\MR^3$ of radius $\pi$ centered at the origin.
In addition, the diametrically opposite points of the bounding sphere
$\MS^2_\pi(0)=\pl\bar\MB^3_\pi(0)$ must be identified, since they correspond to
the same rotation.

The change of variables $n(x)\mapsto\boldsymbol{\om}(x)$ is not a
parameterization in the strict sense of the word. Each value of the
rotation-angle field uniquely defines the $n$-field by formula (\ref{enpara}),
but the converse statement is not true for two reasons. First, the $n$-field
does not define the orthogonal matrix $S$ uniquely, because equality
(\ref{enpara}) does not change if it is multiplied (at every point $x$) by an
arbitrary orthogonal matrix corresponding to rotations around the vector $n(x)$
itself. Second, infinitely many elements of the Lie algebra $\Gs\Go(3)$ are
mapped to the same element of $\MO(3)$. It is the ambiguity of the ``map''
$n(x)\mapsto\boldsymbol{\om}(x)$ that we study in the present section.

The full rotational group consists topologically of two connected components:
$\MO(3)=\MS_+\cup\MS_-$, where $\MS_+$ and $\MS_-$ are the sets of orthogonal
matrices with positive and negative determinants, respectively. The component
$\MS_+$ is the Lie subgroup of special orthogonal matrices
$\MS_+\approx\MS\MO(3)\subset\MO(3)$ (the connected component of unity). The
component $\MS_-$ is a coset of element: $\MS_-=\MS_+g$, where $g$ is any
element in $\MS_-$, for example, $\MS_-=\MS_+(-\one)$, $-\one\in\MS_-$ being
the diagonal $3\times3$ matrix $\diag(-1,-1,-1)$. The Lie algebra $\Gs\Go(3)$
itself is a three-dimensional vector space $\Gs\Go(3)\approx\MR^3$. The
exponential map $\Gs\Go(3)\to\MO(3)$ is surjective because the rotation group is
compact. At the same time, the map $\Gs\Go(3)\to\MO(3)$ is not one-to-one
because infinitely many elements of the algebra are mapped to the same element
of the group.

An explicit parameterization of an orthogonal matrix from the $\MS\MO(3)$
component by the rotation-angle field is
\begin{equation}                                                  \label{elsogt}
  S_{i}{}^j=\dl_i^j\cos\om+\frac{(\om\ve)_i{}^j}\om\sin\om
  +\frac{\om_i\om^j}{\om^2}(1-\cos\om)\qquad\in\MS\MO(3),
\end{equation}
where $\om:=|\boldsymbol{\om}|:=\sqrt{\om^i\om_i}$ is the length of the vector
$\boldsymbol{\om}$. Here we use the notation
\begin{equation}                                                  \label{epeans}
  (\om\ve)_i{}^j:=\om^k\ve_{ki}{}^j\qquad\in\Gs\Go(3),
\end{equation}
where $\ve_{ijk}$ is the totally antisymmetric third-rank tensor, $\ve_{123}=1$.

It is easy to verify that there is only one equivalence relation in the Lie
algebra,
\begin{equation*}
  \boldsymbol{\om}\sim\boldsymbol{\om}+2\pi\frac{\boldsymbol{\om}}\om,
\end{equation*}
such that equivalent elements of the Lie algebra are mapped to the same element
of the rotation group $\MS\MO(3)$.

We have found it more convenient to use another parameterization of the elements
of the Lie algebra: $\lbrace\om^i\rbrace\mapsto\lbrace k^i,\om\rbrace$, where
$k=(k^i:=\om^i/\om)$ is the unit vector along the rotational axis, $k^2=1$,
and $\om\in[-\pi,\pi]$ is the rotation angle. The orthogonal matrix
(\ref{elsogt}) in the new variables is
\begin{equation}                                                  \label{elsogp}
  S_{i}{}^j=\dl_i^j\cos\om+k^k\ve_{ki}{}^j\sin\om
  +k_ik^j(1-\cos\om)\qquad\in\MS\MO(3).
\end{equation}
The inverse matrix is obtained by the substitution $k^i\mapsto-k^i$:
\begin{equation}                                                  \label{elsoop}
  S^{-1}_{~\ ~i}{}^j=\dl_i^j\cos\om-k^k\ve_{ki}{}^j\sin\om
  +k_ik^j(1-\cos\om)\qquad\in\MS\MO(3).
\end{equation}

The following equalities are easy to check:
\begin{equation}                                                  \label{ubbcgk}
  n^i=n_0^jS_j{}^i(\om,k)=n_0^jS_j{}^k(\om,k)S_k{}^i(\psi,n),
\end{equation}
where the rotation axis in the last matrix coincides with the vector $n$ and
the angle $\psi$ is arbitrary and may depend on $x$ in a sufficiently smooth
way. The arbitrariness in the choice of $\psi(x)$ corresponds to gauge
transformations.

Indeed, each rotational matrix uniquely defines the vector $n$, but the inverse
statement is not true: vector $n$ does not define a unique $S$. This can be seen
even by counting the number of independent variables: the vector field $n$ has
two independent components due to the condition $n^2=1$, while the
rotational-angle field $\boldsymbol{\om}$ has three independent components. We
will see in what follows that the additional degree of freedom is a gauge one
and can be eliminated by a gauge transformation.

The following statement is the main result of the paper.
\begin{theorem}
Let $n_0$ be a fixed unit vector and $\big(k(x),\om(x)\big)$ and
$\big(k'(x),\om'(x)\big)$ be two sets of smooth fields related by the gauge
transformation
\begin{align}                                                     \label{ujehtr}
  \sin\om'=&\frac{2\sin\dfrac{\raise-.5ex\hbox{$\om$}}2\sin\upsilon
  \big(\cos\dfrac{\raise-.5ex\hbox{$\om$}}2\sin\upsilon\cos\al
  -\cos\upsilon\sin\al\big)}{1-\big(\cos\upsilon\cos\al
  +\cos\dfrac{\raise-.5ex\hbox{$\om$}}2\sin\upsilon\sin\al\big)^2},
\\ \intertext{or}                                                \label{unvhtr}
  \cos\om'=&\frac{1-2\sin^{\!2}\dfrac{\raise-.5ex\hbox{$\om$}}2\sin^{\!2}
  \upsilon-\big(\cos\upsilon\cos\al
  +\cos\dfrac{\raise-.5ex\hbox{$\om$}}2\sin\upsilon\sin\al\big)^2}
  {1-\big(\cos\upsilon\cos\al
  +\cos\dfrac{\raise-.5ex\hbox{$\om$}}2\sin\upsilon\sin\al\big)^2},
\\ \intertext{and}                                                \label{enbdgt}
  k^{\prime i}=&k^i\cos\al+\left(-k^i\cos\dfrac{\raise-.5ex\hbox{$\om$}}2
  \cos\upsilon+n_0^i\cos\dfrac{\raise-.5ex\hbox{$\om$}}2
  +n_0^jk^k\ve_{kj}{}^i\sin\dfrac{\raise-.5ex\hbox{$\om$}}2\right)
  \frac{\sin\al}{\sin\upsilon},
\end{align}
where the angle $\upsilon$ is defined by the equality
\begin{equation}                                                  \label{edhfpo}
  \cos\upsilon:=(n_0,k)
\end{equation}
and $\al(x)\in\MR$ is an arbitrary smooth transformation parameter. Then
formulas (\ref{enpara}) and (\ref{elsogp}) define the same field $n(x)$. Any two
sets of fields $\big(k(x),\om(x)\big)$ and $\big(k'(x),\om'(x)\big)$ that define
the same field $n(x)$ are related by transformation
(\ref{ujehtr})--(\ref{enbdgt}) for some parameter $\al(x)$.
\end{theorem}
\begin{proof}
To prove the theorem, we need a rather cumbersome but elementary construction,
which is illustrated in Fig.~\ref{fangles}. Assume that the rotation takes the
vector $n_0$ to a vector $n\ne n_0$. This rotation does not define the
rotational matrix uniquely, because after the rotation the vector $n$ can be
additionally multiplied by a rotational matrix whose rotation axis $k$ coincides
with $n$ (see (\ref{ubbcgk})). This can be done independently at every point
$x\in\MM$, which corresponds to the gauge $\MU(1)$ freedom
$\psi(x)\mapsto\psi(x)+\al(x)$, where $\al(x)$ is the transformation parameter.
\begin{figure}[hbt]%------------------------------------------------------------
\hfill\includegraphics[width=.6\textwidth]{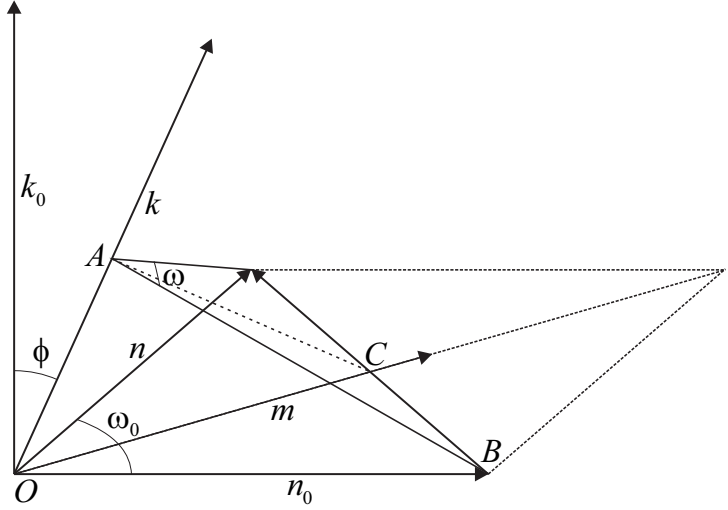}
\hfill {}
\centering\caption{Angle parameterization of rotations.}
\label{fangles}
\end{figure}%-------------------------------------------------------------------

Let us perform calculations. The rotation angle $\om_0$ is minimal if and
only if the rotational axis $k_0$ is perpendicular to the plane passing through
the vectors $n_0$ and $n$. In this case, the unit vector along the rotational
axis is given by the vector product:
\begin{equation}                                                  \label{uvxbcr}
  k_0^i:=\frac{\ve^{ijk}n_{0j}n_k}{\sin\om_0}
\end{equation}
The corresponding rotation angle is defined by the equality
\begin{equation}                                                  \label{ukkdyv}
  \cos\om_0:=(n_0,n):=n_0^in^j\dl_{ij},
\end{equation}
where the parentheses denote the ordinary scalar product in $\MR^3$.

Vector $n$ can be obtained from $n_0$ if and only if the rotation is around an
axis $k$ lying in the plane passing through the vectors $k_0$ and $n_0+n$.
Let $m$ be the unit vector along the sum $n_0+n$. Then its components are
\begin{equation}                                                  \label{uvxfsw}
  m^i:=\frac{n_0^i+n^i}{\sqrt{2\big(1+(n_0,n)\big)}}
  =\frac{n_0^i+n^i}{2\cos\frac\om2}.
\end{equation}
Any unit vector $k$ in the plane $k_0$, $m$ has the form
\begin{equation}                                                  \label{uxbstr}
  k^i=k_0^i\cos\phi+m^i\sin\phi,\qquad\phi\in(-\pi,\pi),
\end{equation}
for some angle $\phi$ in the plane $k_0$, $m$.

Assume that the vector $n_0$ is fixed and we are given values of the variables
$\om_0$, $k_0$ and $\phi$ \big(three independent variables due to the conditions
$k_0^2=1$ and $(n_0,k_0)=0$\big). Then we have to find $\om$ and $k$ to define
the rotation matrix $S_i{}^j(\om,k)$. The vector $k$ is given by (\ref{uxbstr})
with
\begin{equation}                                                  \label{ehsyub}
  n^i=n_0^jS_j{}^i(\om_0,k_0)=n_0^i\cos\om_0+n_0^jk_0^k\ve_{kj}{}^i\sin\om_0.
\end{equation}
To find the angle $\om$, we make the following construction. Consider the right
triangle $ABC$ lying in the plane perpendicular to the vector $k$. Let
$\upsilon$ be the angle between the vectors $n_0$ and $k$ (see (\ref{edhfpo})).
Then
\begin{equation*}
  AB=\sin\upsilon=\sqrt{1-(n_0,k)^2}=\sqrt{1-(n_0,m)^2\sin^{\!2}\phi},
\end{equation*}
where we used equality (\ref{uxbstr}). On the other hand, considering the right
triangle $OBC$, we see that
\begin{equation*}
  BC=\sin\frac{\om_0}2.
\end{equation*}
Consequently,
\begin{equation}                                                  \label{unvbyt}
  \sin\frac\om2=\frac{BC}{AB}
  =\frac{\sin(\om_0/2)}{\sqrt{1-\cos^{\!2}(\om_0/2)\sin^{\!}2\phi}},
\end{equation}
since $(n_0,m)=\cos(\om_0/2)$.

Straightforward calculations yield the formulas
\begin{equation}                                                  \label{ubbsfd}
  \sin\om=\frac{\sin\om_0\cos\phi}{1-\cos^{\!2}(\om_0/2)\sin^{\!2}\phi},\qquad
  \cos\om=\frac{\cos\om_0-\cos^{\!2}(\om_0/2)\sin^{\!2}\phi}
  {1-\cos^{\!2}(\om_0/2)\sin^{\!2}\phi}.
\end{equation}

In view of (\ref{ehsyub}), we have
\begin{equation}                                                  \label{unncbf}
  k^i=k_0^i\cos\phi+\left(n_0^i\cos\frac{\om_0}2
  +n_0^jk_0^k\ve_{kj}{}^i\sin\frac{\om_0}2\right)\sin\phi.
\end{equation}

Thus, formulas (\ref{ubbsfd}) and (\ref{unncbf}) express $\om$ and $k$ in terms
of $\om_0$, $k_0$, and $\phi$ for a fixed vector $n_0$. Moreover, the vector $n$
does not depend on $\phi$:
\begin{equation*}
  n^i=n_0^jS_j{}^i(\om_0,k_0)=n_0^jS_j{}^i(\om,k).
\end{equation*}

When we construct a model in the framework of the geometric theory of defects,
we regard the components of the field $\boldsymbol{\om}(x)$ (three variables)
or, equivalently, $\om(x)$ and $k(x)$ with the additional condition $k^2=1$ as
independent variables. Thus, the number of variables in $\MO(3)$ models
increases from two to three, because the $n$-field does not depend on the field
$\phi(x)$, which was introduced in (\ref{uxbstr}). This field is a gauge
parameter of the $\MU(1)$ transformation $(\om,k)\mapsto(\om',k')$, because
\begin{equation*}
  n^i(x)=n_0^jS_j{}^i(\om,k)=n_0^jS_j{}^i(\om',k'),
\end{equation*}
where the primed fields $\om'$, $k'$ are built for the field
$\phi'(x):=\phi(x)+\al(x)$ with the transformation parameter $\al$ for the same
$\om_0$ and $k_0$. To find an explicit form of the gauge transformations, which
is rather cumbersome, we consider the sequence
$(\om,k)\mapsto(\om_0,k_0,\phi)\mapsto(\om',k')$ of one-to-one transformations.
We find first the transformation $(\om,k)\mapsto(\om_0,k_0,\phi)$ for a given
$\phi$. The rotational matrix (\ref{elsogp}) immediately implies an expression
for the rotation angle $\om_0$:
\begin{equation}                                                  \label{uvvxfs}
  \cos\om_0=(n,n_0)=1-2\sin^{\!2}\dfrac{\raise-.5ex\hbox{$\om$}}2\sin^{\!2}
  \upsilon.
\end{equation}
Straightforward calculations yield an expression for the sine:
\begin{equation}                                                  \label{ubdytr}
  \sin\om_0=2\sin\dfrac{\raise-.5ex\hbox{$\om$}}2
  \sin\upsilon\sqrt{1-\sin^{\!2}\upsilon\sin^{\!2}
  \dfrac{\raise-.5ex\hbox{$\om$}}2}.
\end{equation}
In what follows, we need half-angle expression
\begin{equation}                                                  \label{unvhfy}
  \sin\dfrac{\raise-.5ex\hbox{$\om_0$}}2=\sin\dfrac{\raise-.5ex\hbox{$\om$}}2
  \sin\upsilon,\qquad  \cos\dfrac{\raise-.5ex\hbox{$\om_0$}}2
  =\sqrt{1-\sin^{\!2}\upsilon\sin^{\!2}\dfrac{\raise-.5ex\hbox{$\om$}}2}.
\end{equation}

To find $k_0$, we must compute $\phi$. Multiplying (\ref{uxbstr}) by $n_0$,
we get
\begin{equation}                                                  \label{uvbcfr}
\begin{split}
  \sin\phi=&\frac{\cos\upsilon}
  {\sqrt{1-\sin^{\!2}\upsilon\sin^{\!2}\frac\om2}},
\\
  \cos\phi=&\frac{\cos(\om/2)\sin\upsilon}
  {\sqrt{1-\sin^{\!2}\upsilon\sin^{\!2}\frac\om2}}.
\end{split}
\end{equation}
Now equality (\ref{uxbstr}) implies an expression for the components of $k_0$:
\begin{equation}                                                  \label{uznxbt}
  k_0^i=\frac{k^i\cos\frac\om2-\big(n_0^i\cos\frac\om2
  +n_0^jk^k\ve_{kj}{}^i\sin\frac\om2\big)
  \cos\upsilon} {\sin\upsilon\sqrt{1-\sin^{\!2}\upsilon\sin^{\!2}\frac\om2}}.
\end{equation}

To find an explicit expression for the gauge transformations with parameter
$\al(x)$, we have to substitute the obtained expressions
(\ref{uvvxfs})--(\ref{uznxbt}) into the formulas $\om'=\om'(\om_0,k_0,\phi')$
and $k'=k'(\om_0,k_0,\phi')$ and put $\phi':=\phi+\al$. Explicit formulas are
presented in the statement of the theorem.
\end{proof}

Thus, we have obtained explicit expressions for the gauge $\MU(1)$
transformations $(\om,k)\mapsto(\om',k^{\prime})$ with parameter $\al(x)$. To
check the expressions found, we may assume that the initial state coincides with
the state in which the rotation angle is minimal. Then it is easy to see that
under the substitution $(\om,k,\al)\mapsto(\om_0,k_0,\phi)$ (in this case
$\upsilon=\frac\pi2$) formulas (\ref{ujehtr})--(\ref{enbdgt}) transform into
(\ref{ubbsfd}) and (\ref{unncbf}).

For infinitesimal gauge transformations ($\al\ll1$), formulas
(\ref{ujehtr})--(\ref{enbdgt}) in the linear approximation in $\al$ are
simplified:
\begin{equation}                                                  \label{ubvnfy}
\begin{split}
  \om'=&\om+2\sin\dfrac{\raise-.5ex\hbox{$\om$}}2\cos\om\ctg\upsilon\,\al,
\\
  k^{\prime i}=&k^i+\left(-k^i\cos\dfrac{\raise-.5ex\hbox{$\om$}}2\cos\upsilon
  +n_0^i\cos\dfrac{\raise-.5ex\hbox{$\om$}}2
  +n_0^jk^k\ve_{kj}{}^i\sin\dfrac{\raise-.5ex\hbox{$\om$}}2\right)
  \frac\al{\sin\upsilon}.
\end{split}
\end{equation}

Note that the gauge $\MU(1)$ transformations in the case under consideration are
realized without introducing the gauge field. The $n$-field does not change
under these transformations. Therefore, after the substitution
$n\mapsto\boldsymbol{\om}$ according to (\ref{enpara}), any expression for the
Lagrangian for the $n$-field will be invariant under local transformations
(\ref{ujehtr})--(\ref{enbdgt}) with an arbitrary parameter $\al(x)$. This is not
a very unusual situation. Indeed, we are used to the fact that gauge invariance
arises after the introduction of gauge fields (components of a local connection
form) in the Yang--Mills theory. However, there exist other models with local
invariance. For example, general relativity is invariant under local
transformations (general coordinate transformations), with the metric being not
a gauge field.
%******************************************************************************
\section{Action for the Heisenberg ferromagnet in the geometric theory of
defects}
%*******************************************************************************
In the geometric theory of defects, the $n$-field is parameterized by the
rotation-angle field $\om(x)$ and the unit vector field $k(x)$, $k^2=1$, which
defines the axis of rotation. In addition, one of the three degrees of freedom
is gauge as shown in the previous section. To construct the action, we first
consider the simplest case when the rotation axis $k_0$ is perpendicular to the
vector $n_0$, which defines the orientation of the target space in space-time
(see\ Fig.~\ref{fangles}). In this case, the $n$-field is defined by formula
(\ref{ehsyub}), and the fields $\om_0$ and $k_0$ subject to two conditions
$k_0^2=1$ and $(k_0,n_0)=0$ are independent variables (the gauge freedom is
absent).

For definiteness, we choose the vector $n_0$ along the $z$ axis, i.e.\ set
$n_0=(0,0,1)$. Then the vector $k_0$ lies in the $x,y$ plane and can be
specified in spherical coordinates by one polar angle $\Psi(x)$:
\begin{equation*}
  k_0=(\cos\Psi,\sin\Psi,0).
\end{equation*}
It follows from (\ref{ehsyub}) that the components of the $n$-field are
\begin{equation}                                                  \label{ubvtgj}
\begin{split}
  n^1=&S_3{}^1=k_0^k\ve_{k3}{}^1=~~\sin\Psi\sin\om_0,
\\
  n^2=&S_3{}^2=k_0^k\ve_{k3}{}^2=-\cos\Psi\sin\om_0,
\\
  n^3=&S_3{}^3=\cos\om_0.
\end{split}
\end{equation}
In this case, angular parameterization of the $n$-field is equivalent to the
choice of spherical coordinates in the target space, which is given by the
simple identification $\om_0=\Theta$ and $\Psi=\Phi+\pi/2$. That is, the
Lagrangian of the $\MO(3)$ model is
\begin{equation}                                                  \label{uvcfzt}
  L=\frac12\big(\pl\om_0^2+\sin^{\!2}\om_0\pl\Psi^2).
\end{equation}

Now we consider a gauge model of ferromagnet in a general variables $\om$, $k$.
The form of the rotation matrices (\ref{elsogt}), (\ref{elsogp}) implies
that generally the $n$-field has components
\begin{equation}                                                  \label{uvxclj}
\begin{split}
  n^i(x)=&n_0^jS_j{}^i\big(\om(x),k(x)\big)
  =n_0^i\cos\om+n_0^jk^k\ve_{kj}{}^i\sin\om+k^i(n_0,k)(1-\cos\om),
\\
  n_i(x)=&S^{-1}_{~\ ~i}{}^j\big(\om(x),k(x)\big)n_{0j}
  =n_{0i}\cos\om-k^k\ve_{ki}{}^jn_{0j}\sin\om+k_i(n_0,k)(1-\cos\om).
\end{split}
\end{equation}

Simple straightforward calculations show that the Lagrangian of the Heisenberg
ferromagnet in the new variables has the form
\begin{align}                                                     \label{ubcgvt}
  &L=\frac12(\pl^\al n,\pl_\al n)=
\\                                                                     \nonumber
  &=\frac12\left[1-(n_0,k)^2\right](\pl\om)^2
  -2\left(n_{0i}\cos\frac\om2+n_0^jk^k\ve_{kji}\sin\frac\om2\right)(n_0,k)
  \sin\frac\om2\,\pl^\al\om\pl_\al k^i+
\\                                                                     \nonumber
  &+2\left[(\dl_{ij}\!-n_{0i}n_{0j}){\cos\!}^2\frac\om2\!-n_0^kk^l\ve_{lki}n_{0j}
  \sin\om+\big(\dl_{ij}(n_0,k)^2\!+n_{0i}n_{oj}\big){\sin\!}^2\frac\om2\right]\!
  {\sin\!}^2\frac\om2\,\pl^\al k^i\pl_\al k^j.
\end{align}
This Lagrangian depends on four fields $(\om,k^i)$ with one condition $k^2=1$.
It is invariant with respect to the gauge $\MU(1)$ transformations
(\ref{ujehtr})--(\ref{enbdgt}) with an arbitrary parameter $\al(x)$. The field
$\phi$ from the previous section is transformed in a simple way:
\begin{equation}                                                  \label{ucnbvt}
  \phi\mapsto\phi'=\phi+\al.
\end{equation}
By construction, the Lagrangian (\ref{ubcgvt}) does not depend on $\al$.

As far as we know, the Lagrangian (\ref{ubcgvt}) is a new kind of a gauge model.
The abelian $\MU(1)$ symmetry is realized nonlinearly, and gauge fields are
absent.

Let us rewrite the Lagrangian in terms of the vector $\boldsymbol{\om}=(\om^i)$
(an element of the algebra $\Gs\Go(3)$). The definition of $k$ implies the
equalities
\begin{equation}                                                  \label{ubcdhy}
  k^i:=\frac{\om^i}\om,\qquad\pl_\al k^i=\frac{\pl_\al\om^i}\om
  -\frac{\om^i\pl_\al\om}{\om^2}\qquad (\pl^\al k,\pl_\al k)
  =\frac{(\pl^\al\boldsymbol{\om},\pl_\al\boldsymbol{\om})}{\om^2}
  -\frac{\pl\om^2}{\om^2}.
\end{equation}
The substitution of the obtained expressions in the Lagrangian (\ref{ubcgvt})
yields a more complicated expression
\begin{equation}                                                  \label{ublkid}
\begin{split}
  L=&\frac{\pl\om^2}2\left[1-\frac{\sin^{\!2}\om}{\om^2}-\frac
  {(n_0,\boldsymbol{\om})^2}{\om^2}\left(1-\frac{\sin\om}\om\right)^2\right]+
\\
  &+\frac{(\pl^\al\boldsymbol{\om},\pl_\al\boldsymbol{\om})}{2\om^2}
  \left[\sin^{\!2}\om+\frac{4(n_0,\boldsymbol{\om})^2}{\om^2}\sin^{\!4}
  \frac\om2\right]
  -\frac{2(n_0,\pl_\al\boldsymbol{\om})^2}{\om^2}\sin^{\!2}\frac\om2\cos\om-
\\
  &-\frac{\pl^\al\om(n_0,\pl_\al\boldsymbol{\om})(n_0,\boldsymbol{\om})}{\om^2}
  \left(\sin\om-\frac4\om\sin^{\!2}\frac\om2\cos\om\right)-
\\
  &-\frac{2\pl^\al\om^i\om^jn_0^k\ve_{ijk}}{\om^3}\sin^{\!2}\frac\om2
  \left[(n_0,\pl_\al\boldsymbol{\om})\sin\om+\pl_\al\om(n_0,\boldsymbol{\om})
  \left(1-\frac{\sin\om}{\om}\right)\right].
\end{split}
\end{equation}
The corresponding action depends only on the three fields $\om^i$, which are
varied without any restriction.
%******************************************************************************
\section{Conclusions}
%*******************************************************************************
We have constructed a new gauge parameterization of the Heisenberg ferromagnet
$n$-field by the rotational-angle field $\boldsymbol{\om}$, which is needed in
the geometric theory of defects. In this parameterization, we have three
independent components of the rotational-angle field $\boldsymbol{\om}$ instead
of the two independent components of the $n$-field. We have shown that this
additional degree of freedom is gauge and corresponds to local rotations around
the $n$-field. Explicit formulas of gauge transformations are found. In
addition, any Lagrangian for the $n$-field leads to a gauge
$\MU(1)\simeq\MS\MO(2)$ model in terms of the new variable $\boldsymbol{\om}$.
As an example, we have considered a gauge parameterization of the Heisenberg
ferromagnet. These models do not contain $\MU(1)$ gauge field but are invariant
with respect to local $\MU(1)$ transformations.

{\bf Funding.}

The work was supported in part by the Russian Government Program of Competitive
Growth of Kazan Federal University (Russian Academic Excellence Project
``5--100'').


\begin{thebibliography}{1}

\bibitem{SlaFad88}
L.~D.~Faddeev and A.~A.~Slavnov.
\newblock Gauge Fields: an introduction to quantum theory.
\newblock CRC Press, Roca Raton, USA, Second edition, 2018.

\bibitem{KatVol92}
M.~O. Katanaev and I.~V. Volovich.
\newblock Theory of defects in solids and three-dimensional gravity.
\newblock {\em Ann.\ Phys.}, 216(1):1--28, 1992.

\bibitem{KatVol99}
M.~O. Katanaev and I.~V. Volovich.
\newblock Scattering on dislocations and cosmic strings in the geometric theory
  of defects.
\newblock {\em Ann.\ Phys.}, 271:203--232, 1999.

\bibitem{Katana05}
M.~O. Katanaev.
\newblock Geometric theory of defects.
\newblock {\em Physics -- Uspekhi}, 48(7):675--701, 2005.
\newblock https://arxiv.org/abs/cond-mat/0407469.

\bibitem{Katana13B}
M.~O. Katanaev.
\newblock Geometric methods in mathematical physics.
\newblock Ver. 3, 2016.
\newblock arXiv:1311.0733 [math-ph][in Russian].

\bibitem{Katana17C}
M.~O. Katanaev.
\newblock Chern–-{S}imons term in the geometric theory of defects.
\newblock {\em Phys.\ Rev.\ D}, 96:84054, 2017.
\newblock https://doi.org/10.1103/PhysRevD.96.084054
  https://arxiv.org/abs/1705.07888 [gr-qc].

\end{thebibliography}
\end{document}